\documentclass[aps,prl,preprint,groupedaddress,showpacs,showkeys]{revtex4-1}
\usepackage{amssymb}
\usepackage{latexsym}
\usepackage{amsmath}
\usepackage{amsthm}
\usepackage{amscd}
\usepackage{graphicx}
\newtheorem{proposition}{Proposition}
\newtheorem*{corollary}{Corollary}

\theoremstyle{remark}
\newtheorem{remark}{Remark}

\renewcommand{\c}{\cdot}
\newcommand{\nn}{\nonumber}
\newcommand{\bea}{\begin{eqnarray}}
\newcommand{\eea}{\end{eqnarray}}
\newcommand{\ba}{\begin{array}}
\newcommand{\ea}{\end{array}}

\newcommand{\C}{{{\mathbb C}}}

\newcommand{\al}{\alpha}
\newcommand{\ph}{\varphi}
\renewcommand{\th}{\theta}
\newcommand{\be}{\begin{equation}}
\newcommand{\ee}{\end{equation}}
\newcommand{\bn}{\begin{enumerate}}
\newcommand{\en}{\end{enumerate}}
\newcommand{\fd}{f^\dagger}
\newcommand{\zd}{z^\dagger}

\newcommand{\cpn}{$\C P^{N-1}$}
\newcommand{\bp}{\bar{\partial}}
\newcommand{\p}{\partial}
\newcommand{\bxi}{\bar{\xi}}

\newcommand{\tr}{\mathrm{tr}}

\newcommand{\half}{\tfrac{1}{2}}
\newcommand{\bi}{\begin{itemize}}
\newcommand{\ei}{\end{itemize}}
\begin{document}

% Use the \preprint command to place your local institutional report
% number in the upper righthand corner of the title page in preprint mode.
% Multiple \preprint commands are allowed.
% Use the 'preprintnumbers' class option to override journal defaults
% to display numbers if necessary
%\preprint{}

%Title of paper
\title{The $\mathfrak{su}(2)$ spin $s$ representations via $\mathbb{C}P^{2s}$ sigma
models}
\author{P.P. Goldstein}
\affiliation{National Centre for Nuclear Research, Pasteur St. 7,
02-093 Warsaw, Poland.} \email{piotr.goldstein@ncbj.gov.pl}
\author{A.M. Grundland}
\affiliation{Centre de Recherches Math\'ematiques, Universit\'e de
Montr\'eal, CP 6128, Succ. Centre-Ville, Montr\'eal (QC) H3C 3J7,
Canada,\\ and Department of Mathematics and Computer Science,
Universit\'e du Qu\'ebec, CP 500, Trois-Rivi\`eres, (QC) G9A 5H7,
Canada.} \email{grundlan@CRM.UMontreal.CA}
\author{A.M. Escobar Ruiz}
\affiliation{Centre de Recherches Math\'ematiques, Universit\'e de
Montr\'eal, CP 6128, Succ. Centre-Ville, Montr\'eal (QC) H3C 3J7,
Canada.} \email{escobarr@CRM.UMontreal.CA}
\thanks{The research of AMG and AMER was supported by the Natural Sciences and Engineering Research Council of Canada
operating grant of one of the authors (AMG).}

%\date{Received: date / Accepted: date}
% The correct dates will be entered by the editor

\begin{abstract}
We establish and analyze a new relationship between the matrices
describing an arbitrary component of a spin $s$, where $2s\in
\mathbb{Z}^+$, and the matrices of $\mathbb{C}P^{2s}$
two-dimensional Euclidean sigma models. The spin matrices are
constructed from the rank-1 Hermitian projectors of the sigma
models or from the antihermitian immersion functions of their
soliton surfaces in the $\mathfrak{su}(2s+1)$ algebra. For the
spin matrices which can be represented as a linear combination of
the generalized Pauli matrices, we find the dynamics equation
satisfied by its coefficients. The equation proves to be identical
to the stationary equation of a two-dimensional Heisenberg model.
We show that the same holds for the matrices congruent to the
generalized Pauli ones by any coordinate-independent unitary
linear transformation. These properties open the possibility for
new interpretations of the spins and also for application of the
methods known from the theory of sigma models to the situations
described by the Heisenberg model, from statistical mechanics to
quantum computing.\\
PACS{75.10.Hk, 02.30.Ik, 05.90.+m, 11.10.Lm, 67.57.Lm}
\end{abstract}

\maketitle
\section{Introduction}
\label{intro} The simplest $\mathbb{C}P^{2s}$ sigma models,
$2s\in\mathbb{Z}^+$, were invented by Gell-Mann in 1960 \cite{GM}
and later developed by Callan et al. \cite{CWZ2,CWZ1} to explain
pion lifetime. They have found many other applications since
then, such as
%Carrol,Gross,Polchinski,Pol,Polyakov,David,
\cite{VisPar,Seiberg,Zhit,Rajaraman,Dav,Lan}. In this paper we
consider another possible application of these models. Namely, an
appropriate combination of rank-1 projectors, which are the basic
building blocks of the models, may be a representation of the
$\mathfrak{su}(2)$  algebra in $\mathbb{C}^N$ and thus describe
spin (or isospin) matrices corresponding to the maximal component
of a spin $s=(N-1)/2$. Our objective is to analyze the conditions
which make it possible and suggest its applications.

In Section 2 we summarize the basic information about the
$\mathbb{C}P^{2s}$ models. In Section 3 we show that some linear
combination of the projectors behaves like a component of a spin
vector and is always congruent to the generalized Pauli matrix of
the appropriate size. On the other hand, we present a
counterexample, demonstrating that these combinations of
projectors  cannot always be combinations of the generalized Pauli
matrices. Finally we show that those matrices which actually are
such combinations (or are congruent to them by a
coordinate-independent unitary transformation) satisfy a
propagation equation of a stationary two-dimensional (2D)
Heisenberg model. In Section 4 we recall the results of \cite{CG}
which state that projectors mapping on the directions of the
Veronese vectors always yield spins when combined in the above
way. Possible applications of these results, which include quantum
computing, are mentioned in Section 5.
\section{Basics of the $\mathbb{C}P^{2s}$ sigma models}
\label{sec:1}
The main feature of the nonlinear sigma models in field theory is
that the transformed field admits a very simple effective
Lagrangian density, defined in $\mathbb{C}$, assuming values in
some manifold \cite{WZ}
\be
\mathcal{L}=\partial_\mu \phi^T \partial^\mu \phi, \label{Lphi},
\ee
with appropriate algebraic constraints on the field $\phi$. This
way, the complexity of their dynamics relies on the geometry of
the target space. Such an approach has found many applications,
such as \cite{CWZ2,CWZ1,VisPar,Seiberg,Zhit,Rajaraman,Dav,Lan}.
Even for the very simple $\mathbb{C}P^{N-1}$ models, where the
target is a single complex $N-dimensional$ sphere,
%check with Zakrzewski
their properties are highly nontrivial
\cite{Din,DHZ,WZ,GSZ,WZ1,GY2,GG-rec,GG-inv,GGP,GG-stack}.
These models are the starting point of our research.

As a rule, the domain is parametrized in terms of the complex
variables ${\xi=x+iy} \in\mathbb{C}$, while the target manifold
variables are either vectors  $z$ of a complex unit sphere,
$\zd\!\c\! z=1$, embedded in $\mathbb{C}^N$, or the Grassmannian
homogeneous variables $f$ such that $z=f/(\fd\!\c\! f)^{1/2}$ (the
dagger superscript denotes the Hermitian conjugation). Another
convenient choice of the variables may be projectors $P\in
GL^N(\mathbb{C})$ mapping on the directions of $z$ (and $f$),
namely
\be\label{projectors}
P=z\otimes\zd=\frac{f\otimes\fd}{\fd\c f},
\ee
where $\otimes$ is the tensor product. The last description proves
to be simple and fruitful. The action corresponding to the
Lagrangian density \eqref{Lphi} integrated over the Riemann sphere
$S^2$ (with a constant factor for convenience) becomes
\be\label{action}
\mathcal{A}= \half\int_{S^2} d\xi d\bxi\,\tr\left(\p P\c \bp
P\right),
\ee
under the idempotency condition
\be\label{P-constraint}
P^2-P=0,\quad P^\dagger=P,\quad \tr\,P = 1,
\ee
where $\p$ and $\bp$ are the derivatives with respect to $\xi$ and
$\bxi$ respectively. The Euler-Lagrange (E-L) equations are simply
\cite{WZ}
\be\label{E-L}
[P ,~ \p\bp      P]=0,
\ee
where the square bracket denotes the commutator.
Their solutions satisfying the condition \eqref{P-constraint} may
be obtained by a recurrence procedure. Namely, it was proven in
\cite{Din,DHZ} that all solutions corresponding to the finite action \eqref{action},
expressed in terms of the homogeneous variables $f_k$,
result from a holomorphic solution $f_0$ (any) by
consecutive application of a raising operator
\be\label{f-rec}
f_{k+1}=\mathcal{P}^+ f_k=\left(\mathbb{I}_{2s+1} -
\frac{f_k\otimes \fd_k}{\fd_k\c f_k}\right)\c\p f_k,~~k=0,...,2s,
\ee
where $\mathbb{I}_N$ is the $N\times N$ unit matrix.

The last nontrivial vector $f_{k+1}$ is the antiholomorphic
solution $f_{2s}$ (the action of $\mathcal{P}^+$ on an
antiholomorphic vector obviously yields a zero vector). Similarly
all solutions can be obtained from an antiholomorphic solution by
an analogous lowering operator $\mathcal{P}^-$.

These raising and lowering operators have their counterparts for
projectors \cite{GG-rec,GG-stack}, namely for $k=0,...,2s$
\be\label{rec-proj}
P_{k+1}\!=\! \Pi^+(P_k)= t_{k+1} \p P_k\c P_k\c \bp P_k,\quad P_{k-1}\!=\! \Pi^-(P_k)= t_{k} \bp P_k\c P_k\c \p P_k,
\ee
where the real scalars
\be
t_j(\xi,\bxi)=[\tr(\bp P_j\!\c\! P_j\!\c\! \p P_j)]^{-1}=[\tr(\p
P_{j-1}\!\c\! P_{j-1}\!\c\! \bp P_{j-1})]^{-1}
\ee
for $j=1,...,2s,\text{ while }t_{0}=t_{2s+1}=0$.

The $\mathbb{C}P^{2s}$ sigma models with finite action are
completely integrable \cite{Din}. Furthermore, the E-L equations
can be written in the form of a conservation law
\be\label{cons-law}
\p\,[\bp P,P]+\bp\,[\p P,P]=0,
\ee
which shows that a total differential can be constructed out of
the commutators \eqref{cons-law}. The integral of the total
differential over any contour $\gamma$ (with the constant of
integration ensuring tracelessness) is an antihermitian immersion
function $X_k(\xi,\bxi)$ of a two-dimensional (2D) surface in a
$\mathfrak{su}(N)$ Lie algebra. Moreover, these immersion
functions can explicitly be expressed as linear combinations of
the projectors $P_k,~k=0,...,2s$, \cite{GY2}. The definition and
the explicit form of the immersion functions are
\be\label{XfromP}
X_k=i\int_{\gamma}\left(-[\p P,P]d\xi+[\bp
P,P]d\bar{\xi}\right)=-i\left(P_k+2\sum\limits_{j=0}^{k-1}P_j\right)+\frac{i(1+2k)}{2s+1}\mathbb{I}_{2s+1}.
\ee
They describe 2D soliton surfaces whose conditions of immersion
are the E-L equations \eqref{E-L}. These surfaces have no common
points, except for $\mathbb{C}P^1$, where the only two surfaces
$X_0$ and $X_1$ coincide \cite{GG-conf}.

 The immersion function matrices $X_k$ span a Cartan subalgebra of the $\mathfrak{su}(N)$ algebra.

From the fact that the $P_k$ are mutually orthogonal projectors of
rank 1, it follows that each of them has only one nonzero
eigenvalue equal to 1 and together they constitute a partition of
unity. Hence an appropriate linear combination of these projectors
may have any required set of eigenvalues. A matrix corresponding
to a component of a spin vector, say $S^z$, has eigenvalues $-s,
-s+1,...,s$ with $s$ being a positive integer or half-integer
$(2s\in\mathbb{Z}^+)$, with the exception of the trivial case
$s=0$. This way, a spin matrix may be constructed from
$\mathbb{C}P^{2s}$ projectors as
\be\label{SpinFromP}
S^z = \sum\limits_{k=0}^{2s}(k-s)P_k.
\ee
Up to a constant factor, this combination is equal to the sum of
the immersion functions of the disjoint soliton surfaces $X_k$
\be
S^z= (-i/2)\sum_{k=0}^{2s} X_k,
\ee
which can provide an interpretation of the spin as a composite
phenomenon.

\noindent We will refer to the matrices which may describe
components of spins as \textit{spin matrices}.
%Text with citations \cite{RefB} and \cite{RefJ}.
%\subsection{Subsection title}
\section{Spin matrices}
\label{sec:2}
%as required. Don't forget to give each section
%and subsection a unique label (see Sect.~\ref{sec:1}).
Before proceeding to our main results, we summarize some
properties of the spin matrices associated with the $CP^{2s}$
models expressed in terms of rank-1 Hermitian projectors $P_k$

%Before discussing the conditions which the $S$ matrices \eqref{SpinFromP} have to satisfy to represent spins of a physical system, we summarize their basic properties.
\subsection{Properties of the spin matrices}

Note that all the discussed properties of the spin matrices $S^z$
will follow from the defining relation \eqref{SpinFromP} and from
the fact that the Hermitian matrices $P_k$ map onto
one-dimensional subspaces of $\mathbb{C}^{N}$.

{\bf\emph{Property 1}}. If a Hermitian rank-1 projector $P_k$ maps
onto a one-dimensional subspace of $\mathbb{C}^{N}, ~N=2s+1$, then
the trace and the rank of the spin matrix $S^z$ \eqref{SpinFromP}
are
%{\setlength{\mathindent}{0pt}
\begin{equation}\label{bigproperty1}
\qquad  \mbox{tr}\,S^z \ = \  0 \ , \qquad \mbox{rank}\,S^z \ = \
\Bigg\{\begin{array}{ll}
    N \qquad \quad \ \ \, \mbox{for}\ N\,=\,2n \ ,  \\
    N-1 \qquad \mbox{for}\ N\,=\,2n+1 \  ,
    \end{array} \qquad n\in\mathbb{Z}^+ \ .
\end{equation}

\noindent The fact that, for even $N$, the ranks of the matrices
$S^z$ are equal to $N$, while for odd $N$, the ranks of the
matrices $S^z$ are equal to $N-1$ is due to the presence of a zero
eigenvalue corresponding to the eigenvector $P_k$, where
$k=s=\tfrac{1}{2}(N-1)$.

{\bf \emph{Property 2}}. The quadratic form corresponding to the
Killing form in $\mathfrak{su}(N)$ for the matrix $S^z$ is
constant and may be defined as

\begin{equation}\label{Ssquare}
 \langle\,S^z,\,S^z\,\rangle \ = \ \frac{1}{N}\,\mbox{tr}(S^z\cdot S^z) \ = \ \frac{N^2-1}{12} \ = \ \frac{s(s+1)}{3},
\end{equation}
which corresponds to the expected value for the length of one
(say $z$) component of a quantum-mechanical spin vector.

{\bf\emph{Property 3}}. The spin matrix $S^z$ satisfies an
algebraic condition determined by the characteristic polynomial
corresponding to the eigenvalue problem
\begin{equation}\label{EES}
  (\, S^z \ - \ \lambda_j\,\mathbb{I}_N  \,)\,P_j \ =\  \mathbf{0}  \ .
\end{equation}
has the form
\begin{equation}\label{CPo}
  F(S^z)\ \equiv \ \prod_{k=0}^{N-1}(\, S^z \ -  \ s\,I_N  \,) \ = \ \mathbf{0}  \ ,
\end{equation}
where $I_N$ is the $N\times N$ identity matrix. As all eigenvalues are different, \eqref{CPo} is the minimal polynomial.

{\bf\emph{Property 4}}. According to our earlier result (eq. (47)
of \cite{GG-stack}) if rank-1 projectors $P_k$ satisfy the E-L
equations \eqref{E-L}, then any linear combination of these
projectors is also a solution of \eqref{E-L} (this fact is
nontrivial, because the E-L equations are nonlinear). Hence the
spin matrices $S^z$ \eqref{SpinFromP} also satisfy the same E-L
equations \eqref{E-L}, except that the constraint is \eqref{CPo}
(instead of $P^2=P$).

{\bf\emph{Property 5}}. It follows from Property 4 that the spin
matrices $S^z$ are conditional stationary points of the same
action integral \eqref{action} as the projectors, but the
condition is \eqref{CPo} (rather than $P^2=P$).

\subsection{$S^z$ matrices as spins}

A 3-dimensional (3D) basis for spin matrices describing a system
of spin $s=(N-1)/2$ are three $N\times N$ generalized Pauli
Hermitian matrices whose elements read \cite{Merz}
\bea\label{genPauli}
&&\left(\sigma^x\right)_{mn}\ =\ (\delta_{m,n+1}+\delta_{m+1,n})\sqrt{s(m+n+1)-m\,n}\nn\\
&&\left(\sigma^y\right)_{mn}\ =\ i\,(\delta_{m,n+1}-\delta_{m+1,n})\sqrt{s(m+n+1)-m\,n}\nn\\
&&\left(\sigma^z\right)_{mn}\ =\ 2\,(s-m)\delta_{mn},
\eea
where $0\le m,n\le N-1$
These matrices generate an irreducible representation of $\mathfrak{su}(2)$ in $\mathbb{C}^N$.

We know from the proof of \cite{CG} that a special role is played
by the solutions of the $\mathbb{C}P^{2s}$ sigma models which stem
from the Veronese sequence of holomorphic functions
\be\label{holoVero}
f_0(\xi,\bxi)=\sum_{j=0}^{2s}\binom{2s}{j}^{1/2}\xi^j.
\ee
The functions $f_0$ \eqref{holoVero} and $f_k,~k=1,...,2s$,
obtained from $f_0$ by the recurrence formulae \eqref{f-rec},
define the corresponding rank-1 projectors by means of
\eqref{projectors}. The spin matrix $S^z$ obtained from these
projectors \eqref{SpinFromP} is tridiagonal and can be uniquely
decomposed into a linear combination of matrices
$\sigma^x,\sigma^y, \sigma^z$ \eqref{genPauli}. These solutions
will be discussed in detail in the next section. Unfortunately,
not all spin matrices \eqref{SpinFromP} have such a property.\\
\textit{Example} A simple counterexample may be constructed from
the recurrence formulae applied e.g. to the following holomorphic
vector of a $\mathbb{C}P^{2s}$
\be
f_0(\xi)=(1,\xi,\xi^2,...\xi^{2s}).
\ee
The resulting spin matrix $S^z$ is obtained from the projectors
$P_k$ \eqref{projectors} as their linear combination
\eqref{SpinFromP} where the projectors follow from the recurrence
formulae \eqref{rec-proj} applied to $P_0=f_0\otimes
\fd_0/(\fd_0\c f_0)$. This matrix does not have to be tridiagonal,
whereas any combination of the diagonal $\sigma^z$ and tridiagonal
$\sigma^x$, $\sigma^y$ has to be tridiagonal as its components
are. In the simplest case of $\mathbb{C}P^2$, the $3\times 3
~~S^z$ matrix
%reads
%\be
%\left (
%  \begin {array} {ccc}
%                    \frac {\xi^2\bxi^2} {\xi^2\bxi^2 + 4\xi\bxi +
%                   1} - \frac {1} {\xi^2\bxi^2 + \xi\bxi +
%                    1} &\bxi\left(-\frac {2\xi\bxi} {\xi^2\bxi^2 +
%                    4\xi
%                    \bxi + 1} - \frac {1} {\xi^2\bxi^2 + \xi\bxi +
%                    1} \right) & - \frac {3\xi\bxi^3} {\left(
%                    \xi^2\bxi^2 + \xi\bxi + 1 \right)\left (\xi^2\bxi^2 + 4\xi
%                    \bxi + 1 \right)} \\
%                   \xi\left(-\frac {2\xi\bxi} {\xi^2\bxi^2 +
%                    4\xi\bxi +
%                    1} - \frac {1} {\xi^2\bxi^2 + \xi\bxi +
%                    1} \right) &\xi\bxi\left(\frac {4} {\xi^2\bxi^2 \
%+ 4\xi
%                    \bxi + 1} - \frac {1} {\xi^2\bxi^2 + \xi\bxi +
%                   1} \right) &\bxi\left(-\frac {\xi\bxi}
%{\xi^2\bxi^2 + \xi\bxi + 1} - \frac {2} {\xi^2\bxi^2 + 4\xi
%                 \bxi + 1} \right) \\
%          - \frac {3\xi^3\bxi} {\left(\xi^2\bxi^2 + \xi\bxi +
%                1 \right)\left(\xi^2\bxi^2 + 4\xi\bxi +
%               1 \right)} &\xi\left(-\frac {\xi\bxi} {\xi^2\bxi^2 + \
%            \xi\bxi + 1} - \frac {2} {\xi^2\bxi^2 + 4\xi\bxi +
%            1} \right) &\frac {1} {\xi^2\bxi^2 + 4\xi\bxi +
%      1} - \frac {\xi^2\bxi^2} {\xi^2\bxi^2 + \xi\bxi
%      + 1}
%    \end {array}
%    \right)
%\ee
has a nonzero element
\be\label{extra-trid}
(S^z)_{13}= - 3\xi\bxi^3 \left[\left(
                    \xi^2\bxi^2 + \xi\bxi + 1 \right)\left (\xi^2\bxi^2 + 4\xi
                    \bxi + 1 \right)\right]^{-1}.
\ee
Obviously, $(S^z)_{31}$ is also nonzero as its complex conjugate.
Hence the matrix is not tridiagonal.
\\
On the other hand, it is evident that any diagonalizable $N\times
N$ matrix having the proper eigenvalues is congruent to
$\sigma^z/2$ (and also to $\sigma^x/2$ or $\sigma^y/2$, by
different congruency transformations). The spin-like linear
combination of projectors \eqref{SpinFromP} is Hermitian, so it is
diagonalizable by a unitary matrix. Let $U$ be the unitary
diagonalizing matrix for $S^z$ \eqref{SpinFromP}, both $S^z$ and
$U$ being functions of $(\xi,\bxi)$. Then
\be\label{UnitTransf}
\text{(a)}~~\sigma^z=2\, U^{-1}\cdot S^z\cdot U,\qquad
\text{(b)}~~S^z= \half\, U\cdot \sigma^z\cdot U^{-1}.
\ee
If  $U$ acts on $\sigma^x$ and $\sigma^y$ in the same way, we
obtain three matrices which constitute a basis for another
irreducible representation of $\mathfrak{su}(2)$ in $\mathbb{C}^N$
(it is straightforward to show that they span a Lie subalgebra of
$\mathfrak{su}(N)$).

In the context of (\ref{UnitTransf}b), the whole dynamics of the
spin matrix lies in the unitary transformation $U$. On the other
hand, in some situations, we can analyze the dynamics of spin
matrices without referring to the transformation
\eqref{UnitTransf}.

Let us start with the matrices $S^z$ which are combinations of the
generalized Pauli matrices \eqref{genPauli}
\be\label{LinCom}
S^z(\xi,\bxi) =
\alpha^x(\xi,\bxi)\sigma^x+\alpha^y(\xi,\bxi)\sigma^y+\alpha^z(\xi,\bxi)\sigma^z,\quad
(\alpha^x)^2+(\alpha^y)^2+(\alpha^z)^2=1/4.
\ee
 The coefficients in \eqref{LinCom} are the coordinates of the
spin vector in the basis \eqref{genPauli}. For such matrices $S^z$
we have
\begin{proposition}
Let $\bm{\al}$ be a vector whose components are the
$\al^x,~\al^y,~\al^z$ of equation \eqref{LinCom}. Then it
satisfies the equation
\be\label{DynEq}
\bm{\alpha}\times \bm{\alpha}_{\xi\bxi}=0,
\ee
where $\times$ denotes the usual vector product in $\mathbb{C}^3$.
The vector  $\bm{\al}$ is subject to the normalization condition
\be\label{alpha-constraint}
4\,\bm{\al}\c \bm{\al}-1=0.
\ee
Equation \eqref{DynEq} is a counterpart of the E-L equations
\eqref{E-L} in terms of the vector $\bm{\al}$.
\end{proposition}
\begin{proof}
According to Property 4, the spin matrices $S^z$ satisfy the E-L
equations \eqref{E-L}.
 Substituting
\eqref{LinCom} into those equations, we obtain the coordinates of
\eqref{DynEq} in the basis \eqref{genPauli}. \qed
\end{proof}
Equation \eqref{DynEq}, together with the constraint
\eqref{alpha-constraint} describe stationary states of the 2D
Heisenberg model (see Appendix).

The E-L equations for the $S^z$ matrices follow from the same
action integral \eqref{action} as the equations for the projectors
\eqref{E-L} determining the conditional stationary point of the
action integral \eqref{action} under the condition \eqref{CPo}.
Similarly, the spin-dynamic equations \eqref{DynEq} can be derived
as conditional stationary points of the action integral over the
Riemann sphere
\be\label{alpha-action}
\mathcal{A}_{\al} = \int_{S^2} d\xi\,
d\bxi\,\left[\bm{\alpha}_\xi\c\bm{\alpha}_{\bxi} -
\mu(\xi,\bxi)\left(4\,\bm{\al}\c \bm{\al}-1\right)\right].
\ee
The Lagrange multiplier $\mu=\mu^\dagger\in Aut(\mathbb{C})$ in
the action integral has been introduced
to comply with the constraint \eqref{alpha-constraint}.
Then under the variation of the action \eqref{alpha-action}, we
have the following
\begin{proposition}
The spin dynamic equations \eqref{DynEq} are defined by the
stationary points of the action integral \eqref{alpha-action}.
\end{proposition}
\begin{proof}
The proof is straightforward if we take the Frechet derivative of
\eqref{alpha-action} with respect to $\bm{\al}(\xi,\bxi)$.  This
yields
\be\label{eq+cond}
\bm{\al}_{\xi\bxi}-4\,\mu\, \bm{\al}=0.
\ee
On the vector multiplication of both sides by $\bm{\al}$, we
obtain \eqref{DynEq}.\qed
\end{proof}
\begin{corollary}
It is evident from the above Proof that equation \eqref{DynEq} is
merely the necessary condition for $\bm{\al}$ to follow spin
dynamics. It has to be supplemented by the normalization condition
\eqref{alpha-constraint}. For the action \eqref{alpha-action}, we
can get the complete E-L equations by the scalar multiplication of
\eqref{eq+cond} with $\bm{\al}$, which allows us to calculate the
multiplier $\mu$ explicitly from the $\bm{\al}$-normalization
condition \eqref{alpha-constraint}, thus getting
$\mu=\bm{\al}\c\bm{\al}_{\xi\bxi}$. Substituting this value into
\eqref{eq+cond}, we obtain
\be\label{fullE-L}
 \left(\mathbb{I}_3-4\,\bm{\al}\otimes\bm{\al}\right)\c\bm{\al}_{\xi\bxi}=0,
\ee
where $\mathbb{I}_3$ is the 3D identity tensor.
\end{corollary}

%The vector form of a spin matrix which does not constitute a
%linear combination of $\sigma^x,\sigma^y, \sigma^z$ is not unique.
%The diagonalization of such an $S^z$ leads to $\sigma^z$, which is
%independent of $\xi, \bxi$, therefore the above equations
%\eqref{DynEq} become trivial. On the other hand, it is possible to
%connect a moving frame with the matrix $S^z$ so that $S^z$ is
%congruent to $\sigma^z$ by a unitary transformation $U$, while the
%other two axes are congruent to $\sigma^x$ and $\sigma^y$ by the
%same transformation. A special case arises when the transformed
%matrices are constant. Then they satisfy the same equation
%\eqref{DynEq}, which may be stated as a following proposition
This result can be generalized to the matrices congruent to
$\sigma_x, \sigma_y,\sigma_z$ by a coordinate-independent unitary
transformation, namely
\begin{proposition}\label{prop3}
Let
\be
S^z = \bm{\al}\c\bm{s}:=\sum_{k\in\{x,y,z\}}\,\al^k s^k,
\ee
where $\al^k\in\mathbb{R}$, the Euclidean norm $|\bm{\al}|=1/2$,
while $s^k$ are congruent to $\sigma^k$ by a constant unitary
matrix $U$
\be
s^k = U\c \sigma^k\c U^{-1},~~k\in\{x,y,z\}.
\ee
Then if the commutator $[S^z,~(S^z)_{\xi\bxi}]$ vanishes, the vector
$\bm{\al}$ satisfies equation \eqref{DynEq}.
\end{proposition}
A simple proof follows from the fact that all commutators
$[s^k,\,s^m],~k, m \in \{x,y,z\}$ are congruent to
$[\sigma^k,\,\sigma^m]$ by the same transformation matrix $U$.
\begin{remark}
Proposition \ref{prop3} is trivial for $s=\frac12$ (i.e.
$N=2$) due to the isomorphism of $SU(2)$ with $SO(3)$, which makes
the transformation a rotation of the vector $\bm{\al}$ by a
constant angle. It is nontrivial for higher spins ($N>2$) as the
set of constant $U$ transformations is much richer.
\end{remark}
\begin{remark}
In the general case, the Proposition \ref{prop3} is not true if
the transformation matrix depends on the coordinates. In all the
proven cases, the $S^z$ matrix depends on $\xi,\bxi$ through the
coefficients $\al^x,\al^y,\al^z$. Only two of these coefficients
are algebraically independent (note the normalization condition
\eqref{alpha-constraint}). In general, the system \eqref{E-L} can
have more degrees of freedom.
\end{remark}
%\paragraph{Paragraph headings} Use paragraph headings as needed.
%\begin{equation}
%a^2+b^2=c^2
%\end{equation}

% For one-column wide figures use
%\begin{figure}
% Use the relevant command to insert your figure file.
% For example, with the graphicx package use
%  \includegraphics{example.eps}
% figure caption is below the figure
%\caption{Please write your figure caption here}
%\label{fig:1}       % Give a unique label
%\end{figure}
%
% For two-column wide figures use
%\begin{figure*}
% Use the relevant command to insert your figure file.
% For example, with the graphicx package use
%  \includegraphics[width=0.75\textwidth]{example.eps}
% figure caption is below the figure
%\caption{Please write your figure caption here}
%\label{fig:2}       % Give a unique label
%\end{figure*}
%
% For tables use

\section{Spins from the Veronese vectors}
The E-L equations \eqref{E-L} expressed in terms of the
homogeneous variables $f_k, ~k=0,...,2s$ have the form
\be
\label{E-Lf} \left(\mathbb{I}_{2s+1}-\frac{f_k\otimes
f_k^{\dagger}}{f_k^{\dagger}\cdot f_k}\right)\cdot\left[\p\bp
f_k-\frac{1}{f_k^{\dagger}\cdot f_k}\left((f_k^{\dagger}\cdot\bp
f_k)\p f_k+(f_k^{\dagger}\cdot\p f_k)\bp f_k\right)\right]=0,
\ee
One of the most useful holomorphic solutions (i.e. for $k=0$) is
given by
\begin{equation}\label{Veronese}
f_0=\left(1,\binom{2s}{1}^{1/2}\xi,\dots,\binom{2s}{r}^{1/2}\xi^r,\dots,\xi^{2s}\right)\in\mathbb{C}^{2s+1}\backslash\{\emptyset\}.
\end{equation}
Starting from this solution, a sequence of solutions for
$k=1,...,2s$ may be obtained from the recurrence relations
\eqref{f-rec}. The sequence $\{f_0,f_1,...,f_{2s}\}$ is called the
Veronese sequence \cite{Bolton}.

 In this section we
consider the $\mathbb{C}P^{2s}$ models in which the vectors
$f_0,...,f_{N-1}$ make a Veronese sequence. All these vectors
$f_0,..., f_{N-1}, ~~N=2s+1$, may be expressed in terms of the
Krawchouk orthogonal polynomials \cite{CG}.
\begin{equation}\label{f_k}
(f_k)_{j} = \frac{(2s)!}{(2s-k)!}\left(\frac{-\bxi}{
1+\xi\bxi}\right)^k \binom{2s}{j}^{1/2}\bxi^jK_j(k;p;2s), \qquad
0\leq k,j\leq 2s,
\end{equation}
with the stereographic projection variable $p$
\be
0<p = \frac{\xi_+\xi_-}{1+\xi_+\xi_-}<1.
\ee
Here $(f_k)_j$ is the jth component of the vector $f_k
\in\mathbb{C}^{2s+1}\setminus\{\emptyset\}$ and $K_j(k;p,2s)$ are
the Krawtchouk polynomials for which we use the convention that
for $k=0$
\be
K_j(0;p,2s)=1.
\ee
The Krawtchouk polynomials can be expressed in terms of the
hypergeometric functions \cite{Koornwinder}
\be\label{Kraw-hyper}
K_j(k;p, 2s) = %\prescript{}
_{2}\hspace{-1mm}F_{1}(-j, -k; -2s; 1/p),\quad 0\leq k\leq 2s.
\end{equation}
The element in the $i$-th row and $j$-th column of the rank-1
Hermitian projector $P_k$ as given by \eqref{projectors} has the
form \cite{CG}
\be
(P_k)_{ij} =
\binom{2s}{k}\frac{(\xi\bxi)^{k}}{(1+\xi\bxi)^{2s}}\xi^i\bxi^j\sqrt{\binom{2s}{i}\binom{2s}{j}}K_i(k)K_j(k),
\ee
where we have omitted the dependence of $K_i$ on $p$ and $2s$.

For the Veronese sequence solutions of the $\mathbb{C}P^{2s}$
models, the analytic recurrence relations can be replaced by
simpler algebraic ones. It is convenient to use the combinations
of the $x$ and $y$ components of the spin matrices $S^\pm = S^x\pm
i\,S^y$ and $\sigma^\pm=\sigma^x\pm i\,\sigma^y$, rather than the
components themselves. The spin matrix $S^z$ may be simply
represented by a combination of the diagonal $\sigma^z$ and
tridiagonal $\sigma^+, \sigma^-$, namely \cite{CG}
\be\label{Sz}
 S^z = \frac{1}{2(1+\xi\bxi)}\left[(\xi\bxi-1)\sigma^z-\xi\sigma^--\bxi\sigma^+\right].
\ee
The $\mathfrak{su}(2)$ commutation relations
\be\label{commut}
[S^z,~S^\pm] =\pm S^\pm\, ,\qquad  [S^+,~S^-]=2S^z\,
\ee
(identical to the relations satisfied by the respective
combinations $\sigma^\pm$), suggest the following form of the
components $S^+$ and $S^-$
\bea\label{Spm}
&&S^+ = \frac{1}{2(1+\xi\bxi)}\left(2\bxi\,\sigma^z+\bxi^2\,\sigma^+ -\sigma^-\right),\nn\\
&&S^- = (S^+)^\dagger =
\frac{1}{2(1+\xi\bxi)}\left(2\xi\,\sigma^z-\sigma^+
+\xi^2\,\sigma^- \right)
\eea
(our $\sigma^z$ and $\sigma^\pm$ are twice as large as those of
\cite{CG} to comply with their notion as the generalized Pauli
matrices \eqref{genPauli}).

 It is easy to check that the components of the spin
$S^z,~ S^\pm$ indeed satisfy the commutation relations
\eqref{commut}. Moreover \cite{CG}, $S^\pm$ play the role of the
creation and annihilation operators for $f_k$, namely
\bea
S^+ f_k &=& -(1+\xi\bxi)f_{k+1},\\
S^- f_k &=& \frac{k(k-1-2s)}{1+\xi\bxi}f_{k-1},
\eea
where by convention $f_{-1}=f_{2s+1}=\mathbf{0}$ (see \cite{CG}
for the proof ).

Similarly, we get the algebraic recurrence relations for the projectors, namely
\bea
P_{k+1} &=&
\frac{S^+P_kS^-}{\tr(S^+P_kS^-)},~~\text{for}~k=0,...,2s-1\nn\\
P_{k-1} &=& \frac{S^-P_kS^+}{\tr(S^-P_k
S^+)},~~\text{for}~k=1,...,2s.
\eea
Consequently, the algebraic recurrence relations for the immersion
functions $X_k$ satisfy the algebraic conditions
\begin{align}
&
X_{k+1}=X_k-i\left(\frac{S^+P_kS^-}{tr(S^+P_kS^-)}+P_k-\frac{2}{2s+1}\mathbb{I}_{2s+1}\right)%\in \mathfrak{su}(2s+1)
,\\ &
X_{k-1}=X_k+i\left(\frac{S^-P_kS^+}{tr(S^-P_kS^+)}+P_k-\frac{2}{2s+1}\mathbb{I}_{2s+1}\right)% \in \mathfrak{su}(2s+1)
.
\end{align}

 The algebraic
recurrence relations allow us to recursively construct the
Veronese sequence of solutions $f_k$ (or rank-1 projectors $P_k$)
from the holomorphic solution $f_0$ (or $P_0$) in a simpler way
than the analytic relations \eqref{f-rec}.

The spin matrices corresponding to the Veronese vectors look
particularly simple if we express them in terms of the spherical
coordinates of the vector $\bm{\al}$ from \eqref{LinCom}. Let
\be\label{spherical}
\al^z =\frac12 \cos\th, \quad \al^x=\frac12 \sin \th \cos\ph,\quad
\al^y=\frac12 \sin\th \sin\ph,
\ee
where, as usually, $0\le\th\le\pi, ~0\le\ph<2\pi$. Then a
straightforward calculation leads to
\be\label{Vero-sph}
\th = 2 \arctan|\xi|,\quad \ph =-\arg\xi~(\mathrm{mod}\,2\pi).
\ee
Thus the angle between the vector $\bm{\al}$ and the $z$-direction
depends on the modulus  $|\xi|$ only, while a change in the phase
of $\xi$ yields an identical rotation of the spin vector about the
$z$ axis.

%\begin{table}
% table caption is above the table
%\caption{Please write your table caption here}
%\label{tab:1}       % Give a unique label
% For LaTeX tables use
%\begin{tabular}{lll}
%\hline\noalign{\smallskip}
%first & second & third  \\
%\noalign{\smallskip}\hline\noalign{\smallskip}
%number & number & number \\
%number & number & number \\
%\noalign{\smallskip}\hline
%\end{tabular}
%\end{table}

\section{Possible applications}
The results presented here complement our previous studies
\cite{GG-conf,GG-rec,GG-inv,GGP,GG-stack} on the theory of 2D
Euclidean sigma models. The inclusion of the $\mathfrak{su}(2)$
spin-s representations may have a significant impact on many
problems with physical applications. The equations \eqref{DynEq}
describe stationary states of the 2D Heisenberg model (see
Appendix). The connection which we have found describes this model
of ferro-, antiferro- and ferrimagnets in terms of
$\mathbb{C}P^{2s}$ sigma models, interpreting it as their special
case. The complete integrability of the $\mathbb{C}P^{2s}$ sigma
models provides a useful tool for solving problems of these
magnetic materials, e.g. \cite{Levanyuk}. In particular, the
Veronese solutions of the $\mathbb{C}P^{2s}$ sigma models allow us
to explicitly build the corresponding spin fields.

Our description of the spin-$s$ matrix in terms of the
$\mathbb{C}P^{2s}$ projectors is an example of representing an
array of functions $\mathbb{C}$ into $\mathbb{C}$ as a finite sum
of orthogonal vectors. In the special case of the Veronese
solutions, the vectors were sequences of Krawtchouk polynomials.
Such a representation using the Fourier-Krawtchouk transformation
was recently introduced in \cite{Stob} to achieve a quantum
information processing in constant time. Moreover, this constant-time
signal-evolution analysis works on finite strings with arbitrary length \cite{Stob,Wolf}. The
transformation represents the transformed function limited to a
finite interval in terms of the Krawtchouk polynomials multiplying
the Fourier variable $\exp[-i \frac{\pi}{2}(l-k)],
~k,l\in\mathbb{Z}$. It has its 2D counterpart in splitting the
spin component $S^z$ into the rank-1 projectors proportional to
products of two Krawtchouk polynomials with $\xi\bxi$ in consecutive powers.
Our more general scheme encompasses such a possibility. It is very likely that a discretization (sampling), followed by an appropriate
transformations of this kind, may be suitable for efficient
quantum computations. Further it is promising for possible
applications to digital image processing, in medical image
reconstruction and recognition
\cite{Stob,Nielsen,Zellinger,Yap,Gautam}.
\begin{acknowledgements}
%If you'd like to thank anyone, place your comments here
%and remove the percent signs.
The research was supported in part by The NSERC of Canada
operating grant of one of the authors (AMG). AMER wishes to
acknowledge and thank the Centre de Recherche Math\'ematiques
(Universit\'e de Montr\'eal) and NSERC of Canada for the financial
support provided for his three-year visit to Montreal.
\end{acknowledgements}

% Authors must disclose all relationships or interests that
% could have direct or potential influence or impart bias on
% the work:
%
\section*{Conflict of interest}
On behalf of all authors, the corresponding author states that
there is no conflict of interest.
%
% The authors declare that they have no conflict of interest.

% BibTeX users please use one of
%\bibliographystyle{spbasic}      % basic style, author-year citations
%\bibliographystyle{spmpsci}      % mathematics and physical sciences
%\bibliographystyle{spphys}       % APS-like style for physics
%\bibliography{}   % name your BibTeX data base

% Non-BibTeX users please use
%\begin{thebibliography}{}

%\begin{footnotesize}

\section*{Appendix}
Consider the 2D Heisenberg model consisting of spins
$\bm{s}_{km}$, situated in the positions
$(x_k,y_m),~k,m\in\mathbb{Z}$ of a rectangular lattice whose cell
size is $a\times b$, i.e. $x_k=k\,a,~y_m=m\,b$, $a,b,\in\mathbb{R^+}$. Its Hamiltonian is
given by
\be\label{HHam}
H=\sum_{k,m} \left[J_1 \bm{s}_{km}\c \bm{s}_{k+1m}+J_2
\bm{s}_{km}\c
\bm{s}_{km+l}+J_3(\bm{s}_{km}\c\bm{s}_{k+1m+1}+\bm{s}_{k+1m}\c\bm{s}_{km+1})\right],
\ee
where $J_1,J_2\in\mathbb{R}$, are the coupling constants along the
$x$ and $y$ directions, respectively, $J_3\in \mathbb{R}$ is the
coupling constant over the cell diagonals; the summation
encompasses all nodes of the lattice.

We go to the continuous limit by defining
$\bm{s}_{km}=\bm{\al(x_k,y_m)}$, and expanding the Hamiltonian to
second order in the lattice constants $a, b$. The first order
terms vanish due to the perpendicularity of the spins to their
first derivatives. The continuous Hamiltonian, up to a constant,
reads
\be\label{HH-2}
H =
a^2(J_1-J_2-2J_3){\bm{\al}_x}^2+b^2(-J_1+J_2-2J_3){\bm{\al}_y}^2+
(J_1+J_2+2J_3)\bm{\al}\c\left(a^2\bm{\al}_{xx}+b^2\bm{\al}_{yy}\right),
\ee
(where the subscripts $x,\,y$ denote differentiation).

The equations defining a conditional stationary point of the
Hamiltonian, with the Lagrange multiplier $\mu(x,y)$ corresponding
to the constraint ${\bm{\al}^2=const}$, read
\be\label{cond-min}
a^2\left(J_2+2J_3\right)\bm{\al}_{xx}+b^2\left(J_1+2J_3\right)\bm{\al}_{yy}-2\mu
\bm{\al}=0
\ee
The substitution
\be\label{subst}
\xi=\frac{x}{a\sqrt{2(J_2+2J_3)}}+\frac{i y}{b\sqrt{2(J_1+2J_3)}}
\ee
yields equation \eqref{eq+cond} and consequently the stationary 2D
Heisenberg equation \eqref{DynEq} with \eqref{alpha-constraint} or
\eqref{fullE-L}. \qed
%
% and use \bibitem to create references. Consult the Instructions
% for authors for reference list style.
%
%\bibitem{RefJ}
% Format for Journal Reference
%Author, Article title, Journal, Volume, page numbers (year)
% Format for books
%\bibitem{RefB}
%Author, Book title, page numbers. Publisher, place (year)
% etc
%\end{thebibliography}
\end{document}